\RequirePackage{fix-cm}
\documentclass[smallextended]{svjour3}       
\smartqed  
\usepackage{graphicx}
\usepackage{amsmath}
\usepackage{amsfonts}
\usepackage{amssymb}
\usepackage{algorithmic,algorithm}
\usepackage{color}

\newcommand{\F}{\mathbb{F}}
\newcommand{\Z}{\mathbb{Z}}

\newcommand{\Aut}{\mbox{Aut}}
\newcommand{\GL}{\mbox{GL}}
 \journalname{myjournal}
\begin{document}

\title{Classification of linear codes using canonical augmentation\thanks{Partially funded by grant number DN 02/2/13.12.2016}
}


\author{Iliya Bouyukliev    \and
        Stefka Bouyuklieva
        }


\institute{Iliya Bouyukliev \at
              Institute of Mathematics and Informatics, Bulgarian Academy of Sciences, P.O.
    Box 323, Veliko Tarnovo, Bulgaria,
              \email{iliyab@math.bas.bg}           
           \and
           Stefka Bouyuklieva \at
           Faculty of Mathematics and Informatics, St. Cyril and St. Methodius University of Veliko Tarnovo, Bulgaria,
           \email{stefka@ts.uni-vt.bg}
           }

\date{Received: date / Accepted: date}

\maketitle

\begin{abstract}
We propose an algorithm for classification of linear codes over different finite fields based on canonical augmentation. We apply this algorithm to obtain classification results over fields with 2, 3 and 4 elements.
\end{abstract}

\section{Introduction}
\label{intro}
The concept of canonical augmentation is introduced by Brandan McKay \cite{McKay}. It is a very powerful tool for classification of combinatorial structures. The main idea is to construct only nonequivalent objects (in our case - inequivalent linear codes) and in this way to have a classification of these objects. The construction is recursive, it consists of steps in which the nonequivalent objects are obtained from smaller objects by expanding in
a special way.  The canonical augmentation uses a canonical form to check the so called "parent test" and considers only objects that have passed the test.

The technique of canonical augmentation is used for classification of special types of codes and related combinatorial objects in \cite{Iliya6,BB38,supersymmetry,vanEupenLizonek,Royle}, etc.
The corresponding algorithms construct objects with the needed parameters recursively starting from the empty set. In this way, to classify all linear $[n,k]$ codes, codes of lengths $1,2,\dots,n$ and dimensions $\le k$ are also constructed in the generation process.

We present an algorithm of the same type but with a special modification which makes it much faster in many cases. Our algorithm expands the matrices column by column but starts from the identity $k\times k$ matrix. So it constructs all inequivalent linear $[n,k]_q$ codes without getting codes of smaller dimensions. Restrictions on the dual distance, minimum distance, etc. can be applied.
The algorithm is implemented in the program \textsc{Generation}, which is the first module of the software package \textsc{QextNewEdition}. On the one hand, this program gives us the possibility to classify linear codes with given parameters over fields with $q$ elements. On the other hand, the program can give families of inequivalent codes with certain properties that can be used for extension in length and dimension from the other modules in the package. These modules are also based on the idea of canonical augmentation, which gives the possibility for parallelization. The program is available on the web-page \verb"http://www.moi.math.bas.bg/moiuser/~data/Software/QextNewEdition"



 The base in the process of rejection of isomorphic objects is the theory for canonical representative of an equivalence class. The main terms and definitions are described in Section \ref{sec:preliminaries}. In Section \ref{sec:algorithms} we present two versions of our algorithm for canonical augmentation - extension of a generator matrix column by column or row by row. The last section is devoted to some results obtained by using our algorithms.

\section{Preliminaries}
\label{sec:preliminaries}

Codes which are equivalent belong to the same equivalence class.
Every code can serve as a representative for its equivalence
class. To construct all inequivalent codes with given parameters means to have one representative of each equivalence class. To do this, we use the concept for a canonical representative, selected
on the base of some specific conditions. This canonical
representative is intended to make easily a distinction between
the equivalence classes.

Let $G$ be a group acting on a set $\Omega$. This action defines an equivalence relation such that the equivalence classes are the $G$-orbits in $\Omega$. We wish to find precisely one representative of each $G$-orbit and therefore we use a so-called canonical representative map.

\begin{definition} {\rm\cite{KO}} A canonical representative map for the
action of the group $G$ on the set $\Omega$ is a function
$\rho:\Omega\rightarrow\Omega$ that satisfies the following
two properties:
\begin{enumerate}
\item for all $X\in\Omega$ it holds that $\rho(X)\cong X$,
\item for all $X, Y\in\Omega$ it holds that $X\cong Y$ implies
$\rho(X) = \rho(Y)$.
\end{enumerate}
\end{definition}

For $X\in\Omega$, $\rho(X)$ is the canonical
form of $X$ with respect to $\rho$. Analogously, $X$ is in
canonical form if $\rho(X)=X$. The configuration $\rho(X)$ is the canonical
representative of its equivalence class with respect to $\rho$. We can take for a canonical representative of one equivalence
class a code which is more convenient for our purposes.

Let $q$ be a prime power and $\F_q$ the finite field with $q$ elements, $\F_q^*=\F_q\setminus\{0\}$.
A linear code of length $n$, dimension $k$, and minimum distance $d$ over $\F_q$ is called an $[n, k, d]_q$ code.
Two linear codes of the same length and dimension are equivalent if one can be obtained from the
other by a sequence of the following transformations: (1) a permutation of the coordinate positions of all codewords; (2) a multiplication of a coordinate of all codewords with a nonzero element from $\F_q$; (3) a field automorphism.

We take $\Omega$ to be the set of all linear $[n,k,\ge d]_q$ codes with dual distance at least $d^\perp$, and $G$ be the semidirect product $(\F_q^*\wr S_n)\rtimes_{\theta} \Aut(\F_q)$ where $\theta:\Aut(\F_q)\to \Aut(\F_q^*\wr S_n)$ is a homomorphism such that $\theta_{\alpha}((z,h))=(\alpha(z),h)$ for all $\alpha\in\Aut(\F_q)$ and $(z,h)\in \F_q^*\wr S_n$ (for more details see \cite{KO}). The elements of $G$ fix the minimum and the dual distance of the codes. Using that $\F_q^*\wr S_n\cong Mon_n(\F_q)$ where $Mon_n(\F_q)$ is the group of the monomial $n\times n$ matrices over $\F_q$, we can consider the elements of $G$ as pairs $(M,\alpha)$, $M\in Mon_n(\F_q)$, $\alpha\in\Aut(\F_q)$. An automorphism of the linear code $C$ is a pair $(M,\alpha)\in Mon_n(\F_q)\rtimes \Aut(\F_q)$ such that $vM\alpha\in C$ for any codeword $v\in C$.
The set of all automorphisms of the code $C$ forms the automorphism group
of $C$, denoted by $\Aut(C)$. For linear codes over a prime field the nontrivial transformations are of types (1) and (2) and a sequence of such transformations can be represented by a monomial matrix over the considered field. For binary codes, the transformations (2) and (3) are trivial and therefore $\Aut(C)$ is a subgroup of the symmetric group $S_n$.

We use one more group action. The automorphism group of the code $C$ acts on the set of coordinate positions and partitions them into orbits. The canonical representative map $\rho$ induces an ordering of these orbits.
 The all-zero coordinates, if there are any, form an orbit which we denote by $O_a$. If the code contains codewords of weight 1 then their supports form one orbit, say $O_b$.  The orbits for the canonical representative code $\rho(C)$ are ordered in the following way:
 $O^{(\rho)}_1$ contains the smallest integer in the set $\{1,2,\ldots,n\}\setminus (O^{(\rho)}_a\cup O^{(\rho)}_b)$, $O^{(\rho)}_2$ contains the smallest integer which is not in the set $O^{(\rho)}_a\cup O^{(\rho)}_b\cup O^{(\rho)}_1$, etc.
      If $\phi:C\to\rho(C)$ then the permutational part $\pi_\phi$ of $\phi$
maps the orbits of $C$ into the orbits of $\rho(C)$. Obviously, $\phi(O_a)=O^{(\rho)}_a$ and $\phi(O_b)=O^{(\rho)}_b$. If $\pi_\phi(O_{i_s})=O^{(\rho)}_s$ then $O_{i_1}\prec O_{i_2}\prec \cdots\prec O_{i_m}$. We call the first orbit $O_{i_1}$ special and denote it by $\sigma(C)$.
If $\{1,2,\ldots,n\}=O_a\cup O_b$ then the code contains only codewords with weights $0$ and $1$, and in this case we do not define a special orbit.

\begin{example}
If we order the codewords in a code
lexicographically and then compare the codes according to a
lexicographical ordering of the vectors obtained by concatenation of the ordered nonzero codewords, we can take the smallest code in any equivalence class as a canonical representative. This type of canonical map is very easy to define but computationally expensive to implement.
Consider the binary code $C$ generated by the matrix
$G_C=\displaystyle\left(\begin{array}{cccc}1&0&1&1\\ 0&1&0&1\end{array}\right)$
in details.
The automorphism group of $C$ is $\Aut(C)=\{ id, (13),(24),(13)(24)\}$.
If $\Omega_C$ is the equivalence class of $C$ then $\Omega_C=\{ C_1,\ldots,C_6\}$,
$C_i=\{0,v^{(i)}_1\prec v^{(i)}_2\prec v^{(i)}_3\}$. We order the codes in $\Omega_C$ in the following way:
$$C_i\prec C_j\iff (v^{(i)}_1,v^{(i)}_2,v^{(i)}_3)\prec (v^{(j)}_1,v^{(j)}_2,v^{(j)}_3).$$
Therefore,
$C=\{0,0101,1011,1110\}\succ C_1=\{0,0011,1101,1110\}$.  Hence the code $C_1$ is the canonical form of $C$, $C_1=\rho(C)$.
The coordinates of $C_1$ are partitoned into two orbits under the action of its authomorphism group, namely $O_1=\{1,2\}\prec O_2=\{3,4\}$. For the code $C$ the special orbit is $\sigma(C)=\{1,3\}$.
\end{example}

To find the canonical form of a code is a time-consuming part of the classification. The most popular algorithm for canonical form is the algorithm in McKay's program \textsc{nauty} \cite{nauty}. We use the algorithm described in \cite{Iliya-aut}. Similarly to \textsc{nauty}, this algorithm gives in addition to canonical form, also generating elements of the automorphism group of the considered code. Note that if the coordinates are previously partitioned according to suitable invariants, the algorithm works much faster.

\section{The algorithms}
\label{sec:algorithms}

Using the concept of canonical augmentation, we developed an algorithm in two variants (that were not implemented in the previous versions of \textsc{Q-Extension} \cite{Q-Extension}).

\subsection{Algorithm 1}
\label{Algorithm_1}

The first algorithm is a canonical augmentation column by column. We are looking for all inequivalent linear codes with length $n$, dimension $k$, minimum distance $\ge d$ and dual distance at least $d^\perp\ge 2$.
Without loss of generality we can consider the generator matrices in the form $(I_k\vert A)$ where $A$ is a $k\times (n-k)$ matrix. To obtain the codes we use a recursive construction starting with the identity matrix $I_k$ which generates the trivial $[k,k,1]_q$ code. In the $i$-th step we add a column to the considered generator matrices of the obtained $[k+i-1,k]_q$ codes but we take only those columns which gives codes of length $k+i$ with minimum distance $\ge d_i=d-(n-k)+i$ and dual distance at least $d^\perp$. A strategy for effective generation of these vectors (columns) is described in \cite{IliyaMaya}. Since $d\le n-k+1$, the minimum distance in the beginning is $\le 1$ (it is equal to 1 as we begin with the trivial code). The codes obtained from a code $C$ in this way form the set $Ch(C)$ and they are called the children of $C$.
We say that the code $\overline{C}\in Ch(C)$ passes the parent test, if the added coordinate belongs to the special orbit $\sigma(\overline{C})$. Moreover, we define an action of the automorphism group $\Aut(C)$ on the set of all vectors in $\F_q^k$ and take only one representative from each orbit. By $Ch^*(C)$ we denote a subset of $Ch(C)$ consisting of the codes constructed by $C$ and the taken representatives.

\begin{algorithm}[ht]
\caption{Canonical augmentation column by column}\label{Alg1}
\begin{algorithmic}[1]
\REQUIRE The trivial $[k,k,1]_q$ code $C_k$
\ENSURE A set $U_{n}$ of linear $[n,k,\ge d]_q$ codes with dual distance $\ge d^\perp$
\STATE $U_{n}=\emptyset$
\STATE Augmentation($C_k$);
\end{algorithmic}
\end{algorithm}

Using some lemmas we will prove the following theorem

\begin{theorem}\label{thm:main1}
The set $U_n$ obtained by Algorithm \ref{Alg1}
consists of all inequivalent $[n,k,\ge d]_q$ codes with dual distance at least $d^\perp$.
\end{theorem}

The main idea is to prove that Algorithm \ref{Alg1} gives a tree of codes with root the trivial code $C_k$. The codes obtained in level $i$ represents all inequivalent $[k+i,k]_q$ codes with minimum distance at least $d_i$ and dual distance at least $d^\perp$. Denote the set of these codes by $U_{k+i}$. We have to prove that all constructed codes in $U_{k+i}$ are inequivalent, and that any $[k+i,k]_q$ code with needed minimum and dual distance is equivalent to a code in this set.

\begin{algorithm}[ht]
\caption{Procedure Augmentation($A$: linear code of dimension $k$)}\label{Aug1}
\begin{algorithmic}[1]
\IF {the length of $A$ is equal to $n$ }
\STATE $U_n:= U_n\cup \{A\}$;
\ELSE
\FOR {all codes $B\in Ch^*(A)$}
\IF {$B$ passes the parent test}
\STATE Augmentation($B$);
\ENDIF
\ENDFOR
\ENDIF
\end{algorithmic}
\end{algorithm}

The first lemma proves that the equivalence test for codes that pass the parent test and are obtained from non-equivalent parent codes is not necessary.

\begin{lemma}\label{Lemma:parent}
If $B_1$ and $B_2$ are two equivalent linear $[n,k,d]$ codes
which pass the parent test, their parent codes are also
equivalent.
\end{lemma}

\begin{proof}
Let $B=\rho(B_1)=\rho(B_2)$ be the canonical representative of the equivalence class of the considered codes.
Since both codes pass the
parent test, then the added column is in the special orbit of both codes, or $n\in\sigma(B_i)$, $i=1,2$. This means that there is a map $\psi$ that maps $B_1$ to $B_2$ and the permutational part of $\psi$ fixes $n$-th coordinate.
Hence $\psi=(M,\alpha)$, $M=\left(\begin{array}{cc}M_1&0\\ 0&\lambda\\ \end{array}\right)\in Mon_n(\F_q)$, $\lambda\in\F_q^*$, $\alpha\in \Aut(\F_q)$, and $(M_1,\alpha)$ maps
the parent code of $B_1$ to the parent code of $B_2$. Hence both parent codes are equivalent.
\end{proof}

\begin{lemma}\label{Lemma:equ-parents}
Let $A_1$ and $A_2$ be two equivalent linear codes of length
$r$ and dimension $k$. Then for any child code $B_1$ of $A_1$ which passes the
parent test, there is a child code $B_2$ of $A_2$, equivalent
to $B_1$, such that $B_2$ also passes the parent test.
\end{lemma}

\begin{proof}  Let $G_1$ be a generator matrix of
$A_1$ in systematic form, and $A_2=\psi(A_1)$, $\psi=(M,\alpha)$, $M\in Mon_{r}(\F_q)$, $\alpha\in\Aut(\F_q)$. Let $B_1$ be the code generated by $(G_1\vert a^T)$, $a\in\F_q^k$, and
$B_2$ be the code generated by the matrix $G_2=\psi(G_1)$ and the
vector $b^T=(a^\alpha)^T$, where $a^\alpha$ is obtained from $a$ by applying the field automorphism $\alpha$ to all coordinates. Extend the map $\psi$ to $\widehat{\psi}=(\left(\begin{array}{cc}M&0\\ 0&1\\ \end{array}\right),\alpha)\in Mon_{r+1}(\F_q)\rtimes \Aut(\F_q)$ so $\widehat{\psi}(v,v_{r+1})=(vM,v_{r+1})^\alpha$.  Then $$(G_1\vert a^T)\left(\begin{array}{cc}M&0\\ 0&1\\ \end{array}\right)\alpha=
(G_1M\vert a^T)^\alpha=(G_2\vert b^T)$$
and $B_2=\widehat{\psi}(B_1)$. Hence the codes $B_1$ and $B_2$ are equivalent and so they have the same
canonical representative $B=\rho(B_1)=\rho(B_2)$.

The code $B_1$
passes the parent test and therefore the added column is in the special orbit. Since
$\phi_1\widehat{\psi}^{-1}(B_2)=\phi_1(B_1)=\rho(B_1)=\rho(B_2)$,
$\phi_2=\phi_1\widehat{\psi}^{-1}$ maps $B_2$ to its canonical form $B$. Since $\phi_2$ acts on the added coordinate in the same way as $\phi_1$, this coordinate is in the special orbit
and therefore the code
$B_2$ also passes the parent test.
\end{proof}

To see what happens with the children of the same code $C$, we have to consider the automorphism group of $C$ and the group $G=Mon_n(\F_q)\rtimes \Aut(\F_q)$ which acts on all linear $[n,k]_q$ codes (for more details on this group see \cite{HP}).
A monomial matrix $M$ can be written either in the form $DP$ or the
form $PD_1$, where $D$ and $D_1$ are diagonal matrices and $P$ is a permutation matrix, $D_1=P^{-1}DP$. The multiplication in the group $Mon_n(\F_q)\rtimes \Aut(\F_q)$ is defined by $(D_1P_1\alpha_1)(D_2P_2\alpha_2)=(D_1(P_1D_2^{\alpha_1^{-1}}P_1^{-1})P_1P_2\alpha_1\alpha_2)$, where $B^{\alpha}$ denotes the matrix obtained by $B$ after the action of the field automorphism $\alpha$ on its elements. Obviously, $(AB)^{\alpha}=A^\alpha B^\alpha$ and $P^\alpha=P$ for any permutation matrix $P$.
Let see now what happens if we take different vectors $a, b\in\F_q^k$ and use them in the construction extending the
same linear $[n,k]_q$ code $C$ with a generator matrix
$G_C$. We define an action of the automorphism group $\Aut(C)$ of the code $C$ on the set of all vectors in $\F_q^k$. 
To any automorphism $\phi\in \Aut(C)$ we can correspond an invertible matrix $A_\phi\in
\GL(k,q)$ such that $G'=G_C\phi=A_\phi G_C$, since $G'$ is another
generator matrix of $C$. Using this connection, we obtain a homomorphism $f
\ : \ \Aut(C) \longrightarrow  \GL(k,q)\rtimes \Aut(\F_q)$, $f(M,\alpha)=(A_\phi,\alpha)$.
We have
\begin{align*}
  G_C\phi_1\phi_2 & =(A_{\phi_1}G_C)\phi_2=(A_{\phi_1}G_C)M_2\alpha_2
    =(A_{\phi_1}G_C)^{\alpha_2}M_2^{\alpha_2}\\
   &=A_{\phi_1}^{\alpha_2}G_C^{\alpha_2}M_2^{\alpha_2}=A_{\phi_1}^{\alpha_2}A_{\phi_2}G_C.
\end{align*}
Hence $A_{\phi_1\phi_2}=A_{\phi_1}^{\alpha_2}A_{\phi_2}$ and so $f(\phi_1\phi_2)=f(\phi_1)f(\phi_2)$, when the operation in the group $\GL(k,q)\rtimes \Aut(\F_q)$ is $(A,\alpha)\circ (B,\beta)=(A^\beta B,\alpha\beta)$.
Consider the action
of $Im (f)$ on the set $\F_q^{k}$ defined by $(A,\alpha)(x)=(Ax^T)^{\alpha^{-1}}$ for
every $x\in \F_q^{k}$.

\begin{lemma}\label{lemma:ab}
Let $a,b\in\F^k_q$. Suppose that $a^T$ and
$b^T$ belong to the same $Im(f)$-orbit, where $a^T$ denotes the
transpose of $a$. Then the $[n + 1, k]_q$ codes with generator matrices
$(G_C \ a^T)$ and $(G_C \ b^T)$ are equivalent and if one of them passes the parent test, the other also passes the test. Moreover, if the codes with generator matrices
$(G_C \ a^T)$ and $(G_C \ b^T)$ are equivalent and pass the parent test, the vectors $a^T$ and
$b^T$ belong to the same $Im(f)$-orbit.
\end{lemma}

\begin{proof}
Let the matrices $(G_C\vert a^T)$ and $(G_C\vert b^T)$  generate the codes $C_1$ and $C_2$, respectively, and
$b^T=(A_\phi a^T)^{\alpha^{-1}}$, where $\phi=(M,\alpha)\in\Aut(C)$. Then
$$\widehat{\phi}(G_C\vert b^T)=(G_CM\vert b^T)^{\alpha}=((G_CM)^{\alpha}\vert (b^T)^{\alpha})=
(A_\phi G \ A_\phi a^T)=A_\phi(G \ a^T),$$
where $\widehat{\phi}=(\left(\begin{array}{cc}M&0\\ 0&1\\ \end{array}\right),\alpha)\in Mon_{n+1}(\F_q)\rtimes \Aut(\F_q)$.
Since $A_\phi(G \ a^T)$ is another generator
matrix of the code $C_1$, both codes are equivalent. Moreover, the permutational part of $\widehat{\phi}$ fixes the last coordinate position, hence if $n+1$ is in the special orbit of $C_1$, it is in the special orbit of $C_2$ and so both codes pass (or don't pass) the parent test.

Conversely, let $C_1\cong C_2$ and both codes pass the parent test. It turns out that there is a
map $\psi=(M_{\psi},\beta)\in G$ such that $\psi(C_1)=C_2$ and $\pi_\psi(n+1)=n+1$ where $\pi_\psi$ is the permutational part of $\psi$. Hence $M_\psi=\left(\begin{array}{cc}M_1&0\\ 0&\mu\\ \end{array}\right)$ and
$$(G_C\vert a^T)M_\psi\beta=(G_C M_1\vert \mu a^T)\beta=(G_C M_1\beta\vert (\mu a^T)^\beta)=A(G_C\vert b^T).$$
It follows that $G_C M_1\beta=AG_C$ which means that $(M_1,\beta)\in\Aut(C)$, and $(\mu a^T)^\beta=Ab^T$, so $a^T=((\mu^{-1})^\beta Ab^T)^{\beta^{-1}}$. Since
$$G(\mu^{-1}M_1,\beta)=(\mu^{-1}GM_1)\beta=(\mu^{-1})^\beta (GM_1)^\beta=(\mu^{-1})^\beta AG,$$
we have $((\mu^{-1})^\beta A,\beta)=f(\mu^{-1}M_1,\beta)$.
Hence $(\mu^{-1}M_1,\beta)\in\Aut(C)$ and $a^T$ and $b^T$ belong to the same orbit under the defined action.
\end{proof}

\emph{Proof of Theorem \ref{thm:main1}:}

The algorithm starts with the trivial $[k,k,1]_q$ code $C_k=\F_q^k$. In this case $\Aut(C_k)=Mon_{k}(\F_q)\rtimes \Aut(\F_q)$ and the group partitions the set $\F_q^k$ into $k+1$ orbits as two vectors are in the same orbit iff they have the same weight. We take exactly one representative of each orbit (instead the zero vector) and extend $I_k$ with these column-vectors. If $d_1=2$, we take only the obtained $[k+1,k,2]_q$ code, otherwise we take all constructed codes and put them in the set $ch^*(C)$. All obtained codes pass the parent test.

Suppose that $U_{k+i}$ contains inequivalent $[k+i,k,\ge d_i]_q$ codes with dual distance $\ge d^\perp$, $d_i=d-n+k+i$, and any code with these parameters is equivalent to a code in $U_{k+i}$.
We will show that the set $U_{k+i+1}$ consists only of inequivalent codes, and
any linear $[k+i+1,k,\ge d_{i+1}]_q$ code is equivalent to a code in
the set $U_{k+i+1}$.

Suppose that the codes $B_1, B_2\in U_{k+i+1}$ are equivalent. Since
these two codes have passed the parent test, their parent codes are also equivalent according to Lemma \ref{Lemma:parent}.
These parent codes are linear codes from the set $U_{k+i}$
which consists only in inequivalent codes. The only option for both codes is to have the same parent. But as we take only one vector of each orbit under the considered group action, we obtain only inequivalent children from one parent code (Lemma \ref{lemma:ab}). Hence $B_1$ and $B_2$ cannot be equivalent.

Take now a linear $[k+i+1,k,\ge d_{i+1}]_q$ code $C$ with a
canonical representative $B$. If $\sigma(C)$ is the special orbit, we can reorder the coordinates of $C$ such that one of the coordinates in $\sigma(C)$ to be the last one. So we obtain a code $C_1$ that is permutational equivalent to $C$ and passes the parent test. Removing this coordinate, we obtain a parent code $C_P$ of $C_1$. Since $U_{k+i}$ consists of all
inequivalent $[k+i,k,\ge d_i]_q$ codes with dual distance $\ge d^\perp$, the parent
code $C_P$ is equivalent to a code $A\in U_{k+i}$. According to Lemma
\ref{Lemma:equ-parents}, to any child code of $C_P$ that passes the parent test, there is a child code of $A$ that also passes the test. So there is a child code $C_A$ of $A$ that passes the test, so $C_A\in U_{k+i+1}$, and $C_A$ is equivalent to $C$.
In this way we find a code in $U_{k+i+1}$ which is equivalent to $C$.

Hence in the last step we obtain all inequivalent $[n,k,\ge d]_q$ codes with the needed dual distance.

\medskip
Our goal is to get all linear $[n,k]_q$ codes with given dual distance starting from the $k\times k$ identity matrix. We can also start with all already constructed $[n'<n,k]_q$ codes to get all $[n,k]_q$ codes with the needed properties. Similar algorithms are developed in \cite{vanEupenLizonek,Royle} but these algorithms start from the empty set and generate all inequivalent codes of length $\le n$ and dimensions $1,2,\dots,k$.

\subsection{Algorithm 2}

The second algorithm is a canonical augmentation row by row. We start from the empty set (or set of already given codes with parameters $[n-i,k-i,d]_q$, $1\le i\le k$) and aim to construct all $[n,k,\ge d]_q\ge d^\perp$ codes. In any step we add one row and one column to the considered generator matrix. In the $i$-th step we extend the $[n-k+i-1,i-1,\ge d]_q$ codes to $[n-k+i,i,\ge d]_q$ codes.

We consider generator matrices in the form $(A\vert I_k)$. If $C$ is a linear $[n-k+s,s,\ge d]_q$ code with a generator matrix $(A\vert I_{s})$, we extend the matrix to $\left(\begin{array}{c|c|l} A& I_{s}&0^T\\
\hline a&0\ldots 0&1\\ \end{array}\right)=
\left(\left.\begin{array}{c} A\\ a\\ \end{array} \right| I_{s+1}\right)$, where $a\in\F_{n-k}$. If our aim is to construct codes with dual distance $d^\perp_k\ge d^\perp$, in the $s$-th step we need codes with dual distance $d_s^\perp\ge d^\perp-(k-s)$. The obtained $[n-k+s+1,s+1,\ge d]_q$ codes with dual distance $\ge d^\perp-(k-s)$ are the children of $C$ and the set of all such codes is denoted by $Ch(C)$. The parent test for these codes is the same as in Algorithm \ref{Algorithm_1}. We take a canonical representative for the dual code of $C$ such that $\rho(C^\perp)=\rho(C)^\perp$. The orbits of $C$ are ordered in the same way as the orbits of $C^\perp$ and the special orbit for both codes is the same. The only difference is that if $C$ is a code with zero coordinates then the orbit consisting of these coordinates coincides with the orbit of $C^\perp$ consisting of the supports of the codewords with weight $1$. As in the previous algorithm, we define a group action but now on the vectors in $\F_q^{n-k}$ and take one representative from each orbit for the construction. The corresponding set of codes is denoted by $Ch^*(C)$. Lemma \ref{Lemma:parent} and Lemma \ref{Lemma:equ-parents} hold in this case, too.

If $(A\vert I_k)$ is a generator matrix of $C$ then $(I_{n-k}\vert -A^T)$ generates $C^\perp$. So in the extension in the $s$-th step the vector $-a^T$ expands the considered generator matrix of $C^\perp$ to give a generator matrix of the extended code $\overline{C^\perp}\in Ch(C^\perp)$.
Moreover, $\Aut(C^\perp)=\{ (D^{-1}P,\alpha)\vert (DP,\alpha)\in\Aut(C)\}$. Therefore, for the action of $\Aut(C)$ on the vectors in $\F_q^{n-k}$, we use the elements of $\Aut(C^\perp)$.
If $\phi=(DP,\alpha)\in \Aut(C)$ then $\phi'=(D^{-1}P,\alpha)\in \Aut(C^\perp)$ and so we have an invertible matrix $B_\phi\in
\GL(n-k,q)$ such that $G'=(I_k\vert -A^T)\phi'=B_\phi (I_k\vert -A^T)$, since $G'$ is another
generator matrix of $C^\perp$. In this way we obtain a homomorphism $f'
\ : \ \Aut(C) \longrightarrow  \GL(n-k,q)\rtimes \Aut(\F_q)$, $f(DP,\alpha)=(B_\phi,\alpha)$.
Then we consider the action
of $Im (f')$ on the set $\F_q^{n-k}$ defined by $(B,\alpha)(x)=(Bx^T)^{\alpha^{-1}}$ for
every $x\in \F_q^{n-k}$. This action is similar to the action defined in Subsection \ref{Algorithm_1}. The proof of the following lemma for an $[n,k]$ code $C$ with a generator matrix $(A\vert I_k)$ is similar to the proof of Lemma \ref{lemma:ab}.

\begin{lemma}\label{lemma:ab2}
Let $a,b\in\F^{n-k}_q$. Suppose that $a$ and
$b$ belong to the same $Im(f')$-orbit. Then the $[n + 1, k+1]_q$ codes with generator matrices
$\left(\left.\begin{array}{c} A\\ a\\ \end{array} \right| I_{k+1}\right)$ and $\left(\left.\begin{array}{c} A\\ b\\ \end{array} \right| I_{k+1}\right)$ are equivalent and if one of them passes the parent test, the other also passes the test. Moreover, if the codes with generator matrices
$\left(\left.\begin{array}{c} A\\ a\\ \end{array} \right| I_{k+1}\right)$ and $\left(\left.\begin{array}{c} A\\ b\\ \end{array} \right| I_{k+1}\right)$ are equivalent and pass the parent test, the vectors $a$ and
$b$ belong to the same $Im(f')$-orbit.
\end{lemma}

The proof that Algorithm 2 gives the set $U_n$ of all inequivalent $[n,k,\ge d]_q$ codes with dual distance $\ge d^\perp$ is similar to the proof of Theorem \ref{thm:main1}, therefore we skip it.


%

\subsection{Some details}

The parent test is an expensive part of the algorithms. That's way we use invariants to
take information about the orbits $\{O_1,\ldots,O_m\}$ after the action of $\Aut(C)$ on the set of coordinate positions.
An invariant of the coordinates of $C$ is a function $f:
N\to\Z$ such that if $i$ and $j$ are in the same orbit with
respect to $\Aut(C)$ then $f(i)=f(j)$, where $N=\{1,2,\dots,n\}$ is the set of the coordinate positions.
The code $C$ and the invariant $f$ define a partition $\pi= \{
N_1,N_2,\dots,N_l\}$ of the coordinate set $N$,  such that
$N_i\cap N_j=\emptyset$ for $i\not =j$, $N=N_1\cup
N_2\cup\dots\cup N_l$, and two coordinates $i,j$ are in the same
subset of $N \iff  f(i)= f(j)$. So the subsets $N_i$ are unions of
orbits, therefore we call them pseudo-orbits. We can use the fact
that if we take two coordinates from two different subsets, for
example $s\in N_i$ and $t\in N_j$, $N_i\cap N_j=\emptyset$, they
belong to different orbits under the action of $\Aut(C)$ on the
coordinate set $N$. Moreover, using an invariant $f$, we can
define a new canonical representative and a new special orbit of $C$ in the following way.
If $f_i=f(j_i)$ for $j_i\in N_i$, $i=1,2,\dots,l$, we can order the pseudo-orbits in respect to the integers $f_i$. We take for a canonical representative a code for which $f_1<f_2<\cdots <f_l$. Moreover, we order the orbits in one pseudo-orbit as it is described in Section \ref{sec:preliminaries}. So the orbits in the canonical representative are ordered according this new ordering. The special orbit for a code $C$ is defined in the same way as in Section \ref{sec:preliminaries} (only the canonical map and the canonical representative may be different).

In the step "\textit{if $B$ passes the parent test}", using a
given generator matrix of the code $B$ we have to calculate
invariants, and in some cases also canonical form and the
automorphism group $\Aut(B)$. Finding a canonical form and the
automorphism group is necessary when the used invariants are not
enough to prove whether the code $B$ pass or not the parent test.
If the code $B$ passes the parents test, the algorithm needs a set
of generators of $\Aut(B)$ for the next step (finding the child
codes).
Description of some very effective invariants and the process of their applications are described in details in \cite{Iliya-aut} and \cite{nauty}.

Similar algorithms can be used to construct linear codes with a prescribed fixed part - a residual code or a subcode.

\section{Results and verification}
\label{sec:results}

We use the presented algorithm implemented in the program \textsc{Generation} to obtain a systematic classification of linear codes with specific properties and parameters over fields with 2, 3 and 4 elements. There are possibilities for different restrictions for the codes in addition to the restrictions on length, dimension, minimum and dual distances. We apply also restrictions on the orthogonality and weights of the codewords in some examples. The calculations took about two weeks on a 3.5 Ghz PC.

 We classify three types of codes, namely self-orthogonal codes over $\F_q$ for $q=2,3,4$,
divisible binary, ternary and quaternary codes, and optimal binary codes of dimension 8.

The results are presented in tables. Tables \ref{table-q2-n27}, \ref{table-q2-n20} and \ref{table-q2-n18} give the number of all inequivalent binary codes with the prescribed property (self-orthogonal with $d\ge 8$ for Table \ref{table-q2-n27}, resp. even codes for the other two tables) of the needed length $n$ and all dimensions from 2 (resp. 3 and 4) to 12 (resp. 10) including the codes with zero columns (dual distance 1). Tables \ref{table-q3-n20all}, \ref{table-q4-n21}, \ref{table-q3-n50all} and \ref{table-q4-n30} present the number of the inequivalent codes of the corresponding type with lengths and dimensions less than or equal to given integers $n$ and $k$, and dual distance at least 2.

\paragraph{Self-orthogonal codes.}
There are a few tables of self-orthogonal codes (see \cite{BBGO,IliyaPatric,supersymmetry}). Here we present classification results that are not given in these tables, namely:
\begin{itemize}
\item Binary self-orthogonal codes. We present classification results for binary self-orthogonal $[27,k\le 12,d\ge 8]$ codes with dual distance $d^\perp\ge 1$ in Table \ref{table-q2-n27}. The codes with dimensions 11 and 12 are optimal as linear codes, and the codes with $k=9$ and 10 are optimal only as self-orthogonal \cite{BBGO}. Moreover, we tried to fill some of the gaps in \cite[Table 1]{BBGO}. We classified the $n$-optimal self-orthogonal $[n,k,d]$ codes (the codes for which no $[n-1,k,d]$ self-orthogonal code exists) with parameters $[35,8,14]$, $[29,9,10]$ and $[30,10,10]$. The number of codes in these cases are 376, 36504 and 573, respectively. Our program shows that no self-orthogonal $[37,10,14]$ and $[36,9,14]$ codes exist which means that the maximal possible minimum distance for self-orthogonal codes with these lengths and dimensions is 12.

\begin{table}
\caption{Binary self-orthogonal $[27,k\le 12,d \ge 8]d^\perp\ge 1$ codes}
\label{table-q2-n27}       
\begin{tabular}{c|cccccc}
\hline\noalign{\smallskip}
k & 2 & 3&4&5&6&7  \\
\noalign{\smallskip}\hline\noalign{\smallskip}
total & 59  & 445 &4615  & 64715  & 959533&8514764 \\
\noalign{\smallskip}\hline\noalign{\smallskip}
k & 8 & 9&10&11&12&  \\
\noalign{\smallskip}\hline\noalign{\smallskip}
total &  21256761 &7030920  &159814  &791   &18 & \\
\noalign{\smallskip}\hline
\end{tabular}
\end{table}


\item  Ternary self-orthogonal codes. The classification results for $[n\le 20,k\le 10,d\ge 6]$ codes are given in Table \ref{table-q3-n20all}. This table supplements \cite[Table 1]{IliyaPatric}.


\begin{table}
\caption{Ternary self-orthogonal codes with $n\le 20$, $k\le 10$, and $d\ge 6$}
\label{table-q3-n20all}       
\begin{tabular}{c|ccccccc}
\hline\noalign{\smallskip}
$n\setminus k$ &4&5&      6&       7&     8&   9&10\\
\noalign{\smallskip}\hline\noalign{\smallskip}
10   &1   &       &        &        &      &    &\\
11   &1   &      1&        &        &      &    & \\
12   &6   &      2&       1&        &      &    & \\
13   &10  &      4&       1&        &      &    & \\
14   &27  &     15&       4&        &      &    & \\
15   &78  &     73&      20&       2&      &    &\\
16   &181 &    312&     121&      11&     1&    & \\
17   &414 &   1466&     885&      86&     2&    & \\
18   &1097&   8103&   10808&    1401&    40&    &\\
19   &2589&  47015&  167786&   45950&  1132&  10& \\
20   &6484& 285428& 2851808& 2121360& 89670& 464& 6 \\
\noalign{\smallskip}\hline
\end{tabular}
\end{table}

\item Quaternary Hermitian self-orthogonal codes. Table \ref{table-q4-n21} shows the classification of the $[n\le 21,\le 6,12]$ codes of this type. These results fill some of the gaps in \cite[Table 2]{IliyaPatric}.
\end{itemize}

\begin{table}
\caption{Quaternary Hermitian self-orthogonal codes with $n\le 21$, $k\le 6$, $d=12$}
\label{table-q4-n21}       
\begin{tabular}{c|ccccc}
\hline\noalign{\smallskip}
$n\setminus k$ & 2 & 3&4&5&6  \\
\noalign{\smallskip}\hline\noalign{\smallskip}
15&                   1&&&\\
16&                   2&         1&&&\\
17&                   3&         4&         1&&\\
18&                    &        45&        12&&\\
19&                    &          &      5673&&\\
20&                    &          &          &  886576&\\
21&                    &         &&                    &   577008\\                                                                                                                                                                                                                                                                                                                                                                                                                                                                                                                                                  \noalign{\smallskip}\hline
\end{tabular}
\end{table}

\paragraph{Divisible codes.}
Divisible codes have been introduced by Ward in 1981 \cite{Ward-1981,Ward-survey}. They are related
to self-orthogonal codes, Griesmer codes and other combinatorial structures. From the divisible codes with given $n$ and $\Delta$ one can obtain infinite families of Griesmer codes \cite{Iliya_dual}.
A linear code $C$ is said to be $\Delta$-divisible for a positive integer $\Delta$ if all its weights are multiples of $\Delta$.
The main case of interest is that $\Delta$ is a power of the characteristic of the base field. All binary self-orthogonal codes are 2-divisible, and doubly-even codes are 4-divisible. Systematic classification of binary doubly even codes is presented in \cite{supersymmetry} because of their relation to Adinkra
chromotopologies. Recently, 8-divisible codes (called also triply even) have been investigated \cite{triply-Munemasa,triply-germans}.
In \cite{triply-germans1}, it is proven that projective triply-even binary
codes exist precisely for lengths 15, 16, 30, 31, 32, $45-51$, and $\ge 60$.

Using the program \textsc{Generation}, we have classified 2 and 4-divisible binary codes, 9-divisible ternary codes and 4-divisible quaternary codes.

\begin{itemize}
\item $q=2$, $\Delta=2$. The numbers of even binary codes with parameters $[n=20,3\le k\le 10,d\ge 6]$ and $[n=18,4\le k\le 12,d\ge 4]$ with dual distance $d^\perp\ge 1$ are presented in Tables \ref{table-q2-n20} and \ref{table-q2-n18}, respectively.

\begin{table}
\caption{Even binary codes with $n=20$, $k\le 10$, $d\ge 6$, $d^\perp\ge 1$}
\label{table-q2-n20}       
\begin{tabular}{c|cccccccc}
\hline\noalign{\smallskip}
k &  3&4&5&6&7&8&9&10  \\
\noalign{\smallskip}\hline\noalign{\smallskip}
total &  516&6718 &119547 &2075823 &18926650&40411393&5709084&1681 \\
\noalign{\smallskip}\hline
\end{tabular}
\end{table}

\begin{table}
\caption{Even binary codes with $n=18$, $4\le k\le 12$, $d\ge 4$, $d^\perp\ge 1$}
\label{table-q2-n18}       
\begin{tabular}{c|ccccccccc}
\hline\noalign{\smallskip}
k &  4&5&6&7&8&9&10&11&12  \\
\noalign{\smallskip}\hline\noalign{\smallskip}
total &  4923& 51398& 434906& 2083739& 3940649& 2172481& 265798& 5598& 30   \\
\noalign{\smallskip}\hline
\end{tabular}
\end{table}

\item $q=2$, $\Delta=4$. A table with classification results for doubly even binary codes of length $n\le 32$ and dimension $k\le 16$ is given in \cite{supersymmetry}. For dimensions 4, 5 and 6 we repeated the results and further filled the gaps with the number of all inequivalent codes with parameters $[32,4,\ge 4]$ (2163 codes), $[31,5,\ge 4]$ (42656 codes), $[32,5,\ge 4]$ (84258 codes), $[31,6,\ge 4]$ (2,374,543 codes), and $[32,6,\ge 4]$ (6,556,687 codes).
\item $q=3$, $\Delta=9$. Table \ref{table-q3-n50all} contains classification results for codes of this type with length $n\le 50$ and dimension $k\le 6$.


\begin{table}
\caption{Divisible ternary codes with $n\le 50$, $k\le 6$, $\Delta=9$}
\label{table-q3-n50all}       
\begin{tabular}{c|ccccc}
\hline\noalign{\smallskip}
$n\setminus k$ & 2 & 3&4&5&6  \\
\noalign{\smallskip}\hline\noalign{\smallskip}
12 & 1 &    &     &     &\\
13 &   &  1 &     &     &\\
18 & 1 &    &     &     &\\
21 & 1 &  1 &     &     &\\
22 &   &  1 &   1 &     &\\
24 & 1 &  1 &   1 &     &\\
25 &   &  1 &   1 &   1 &\\
26 &   &  1 &   1 &   1 &   1\\
27 & 2 &  3 &   3 &   1 &\\
30 & 2 &  4 &   3 &     &\\
31 &   &  2 &   3 &   1 &\\
33 & 1 &  5 &   5 &   3 &\\
34 &   &  2 &   5 &   4 &   1\\
35 &   &  1 &   4 &   4 &   3\\
36 & 4 & 10 &  22 &  13 &   4\\
37 &   &  2 &   7 &  10 &   3\\
38 &   &  1 &   6 &  12 &  10\\
39 & 3 & 15 &  34 &  41 &  23\\
40 &   &  6 &  25 &  40 &  30\\
42 & 2 & 17 &  52 &  44 &  15\\
43 &   &  6 &  32 &  40 &  16\\
44 &   &  2 &  14 &  22 &  17\\
45 & 5 & 31 & 141 & 190 &  72\\
46 &   &  6 &  56 & 122 &  71\\
47 &   &  2 &  29 &  92 &  89\\
48 & 5 & 44 & 297 & 705 & 468\\
49 &   & 15 & 177 & 613 & 596\\
50 &   &  2 &  39 & 217 & 295\\
\noalign{\smallskip}\hline\noalign{\smallskip}
total & 28 & 182&958&2176&1714 \\                                                                                                                                                                                                                                                                                                                                                                                                                                                                                                                                                \noalign{\smallskip}\hline
\end{tabular}
\end{table}

\item $q=4$, $\Delta=4$. Table \ref{table-q4-n30} presents classification results for codes with $n\le 30$ and $k\le 8$. All constructed codes are Hermitian self-orthogonal.
\end{itemize}

\begin{table}
\caption{Divisible quaternary codes with $n\le 30$, $k\le 8$, $\Delta=4$}
\label{table-q4-n30}       
\begin{tabular}{c|ccccccc}
\hline\noalign{\smallskip}
$n\setminus k$ & 2 & 3&4&5&6&7&8  \\
\noalign{\smallskip}\hline\noalign{\smallskip}
5  & 1 &     &      &     &&&\\
8  & 1 &     &      &     &&&\\
9  & 1 &   1 &      &       &    &&\\
10 & 1 &   1 &    1 &       &    &&\\
12 & 2 &   2 &      &       &    &&\\
13 & 2 &   3 &    1 &       &    &&\\
14 & 1 &   5 &    3 &     1 &    &&\\
15 & 1 &   3 &    6 &     2 &     1 &&\\
16 & 4 &   9 &    7 &     2 &       &&\\
17 & 3 &  12 &    9 &     2 &       &&\\
18 & 2 &  18 &   25 &     8 &     1 &&\\
19 & 1 &  14 &   42 &    25 &     6 &     1 &\\
20 & 6 &  34 &   93 &    70 &    22 &     4 &1\\
21 & 5 &  45 &  115 &    75 &    19 &     2 &\\
22 & 3 &  64 &  245 &   131 &    23 &     2 &\\
23 & 2 &  62 &  554 &   398 &    96 &    12 &     1\\
24 & 9 & 123 & 1509 &  1769 &   491 &    79 &     9\\
25 &   & 168 & 3189 &  6890 &  1842 &   334 &    46\\
26 &   &     & 8420 & 18377 &  2691 &   360 &    33\\
27 &   &     &      & 70147 &  4602 &   458 &    34\\
28 &   &     &      &       & 36982 &  3075 &   244\\
29 &   &     &      &       &       & 34180 &  2366\\
30 &   &     &      &       &       &       & 24565\\                                                                                                                                                                                                                                                                                                                                                                                                                                                                                                                            \noalign{\smallskip}\hline\noalign{\smallskip}
total &  45& 564  &14219  & 97897 & 46776 & 38507 & 27299  \\                                                                                                                                                                                                                                                                                                                                                                                                                                                                                                                                                \noalign{\smallskip}\hline
\end{tabular}
\end{table}

\begin{table}
\caption{Binary $n$-optimal codes with dimension 8}
\label{table-q2-k8}       
\begin{tabular}{c|ccccccccccc}
\hline\noalign{\smallskip}
$d$ & 4&6&8&10&12&14&16&18&20&22&24  \\
\noalign{\smallskip}\hline\noalign{\smallskip}
$n(8,d)$ & 13 &17  & 20  & 26&29&33&36&42&45&48&51 \\
\noalign{\smallskip}\hline\noalign{\smallskip}
total & 1&1&1& 563960&73&2&2&$\ge 352798$&$\ge
424207$&1&1  \\
\noalign{\smallskip}\hline
\end{tabular}
\end{table}

\paragraph{Optimal binary codes.} Table \ref{table-q2-k8} contains classification results for $n$-optimal binary linear codes of dimension 8. Let $n(8,d)$ be the smallest integer $n$ for which a binary linear $[n,8,d]$ code exists. We consider $[n(8,d),8,d]$ for even $d\le 24$. The classifications continue the research from \cite{IliyaJaffe} where $n$-optimal codes with dimensions up to $7$ were investigated. The classification results for $d\le 8$ are known but we give them in the table for completeness (see \cite{Jaffe30}). All constructed codes with minimum distance 8, 12, 16 and 24 are doubly even. We give the number of all doubly even $[45,8,20]$ codes obtained in \cite{Kurz46}, and  we conjecture that there are no more codes with these parameters. For $d=18$, we succeeded to classify only the codes with dual distance 2. Starting from the 172 $[40,7,18]$ codes, \textsc{Q-Extension} gives 352798 $[42,8,18]$ codes with $d^\perp=2$.

\paragraph{Verification.} We use two basic approaches to verify our program and the results. The first one is verification by replication. We ran the program to get already published classification results as the classification of doubly-even binary codes from \cite{supersymmetry}, binary projective codes with dimension 6 \cite{Iliya6}, different types of self-orthogonal codes \cite{BBGO,IliyaPatric,triply-germans}, and we obtained the same number of codes.

The second approach is to use enumeration of different types of codes given by theoretical methods (see \cite{book-germans,enum-SO}). For self-orthogonal codes we can also use mass formulae to verify that the constructed codes represents all equivalence classes of the given length \cite{Gaborit}.

\section*{Conclusion}

There are a few programs for classification of linear codes (see \cite{book-germans,Jaffe30,Ostergaard2002}).
Here we propose a new classification algorithm based on canonical augmentation. It is implemented in the program \textsc{Generation}, which is the first module of the software package \textsc{QextNewEdition}.


%
%



\end{document}